\documentclass{elsarticle}

\makeatletter
\def\@seccntformat#1{\csname the#1\endcsname.\quad}

\usepackage[tbtags]{amsmath}
\usepackage{amsfonts}
\usepackage{amssymb}
\usepackage{ifpdf}
\ifpdf
\usepackage{pdfsync}
\fi
\usepackage{mathrsfs}
\usepackage{enumerate}
\usepackage{gastex}
\usepackage{graphics,graphicx}
\usepackage{xcolor}
\usepackage{latexsym}
\usepackage{ifthen,verbatim}
\usepackage{algorithm}
\usepackage{algpseudocode}

\usepackage{comment}
\usepackage{float}

\usepackage{tikz}
\usetikzlibrary{positioning, calc}

\algtext*{EndWhile}
\algtext*{EndIf}
\algtext*{EndFunction}
\algtext*{EndFor}

\newcommand{\N}{\mathbb N}

\newcommand{\ST}{\mathrm{ST}}

\newcommand{\I}[3]{
  \draw [|-|] (#2)  -- ++(#3, 0);
  \node[above]  at ++ ($(#2) + (#3 / 2, 0)  $) {\ensuremath{#1}}
}

\newtheorem {theorem}{Theorem}
\newtheorem {lemma}[theorem]{Lemma}
\newtheorem {proposition}[theorem]{Proposition}
\newtheorem {claim}[theorem]{Claim}
\newproof{proof}{Proof}

\begin{document}

\title{Dynamic Algorithms for Interval Scheduling on a Single Machine}

\author{Alexander Gavruskin}
\author{Bakhadyr Khoussainov}
\author{Mikhail Kokho}
\author{Jiamou Liu}

%\institute{Department of Computer Science, University of Auckland, New Zealand \and
%School of Computing and Mathematical Sciences, Auckland University of Technology, New Zealand\\
%\email{a.gavruskin@auckland.ac.nz, bmk@cs.auckland.ac.nz, m.kokho@auckland.ac.nz, jiamou.liu@aut.ac.nz}}

\maketitle

\begin{abstract}
We investigate dynamic algorithms for the interval scheduling problem. Our algorithm runs in amortised time  $O(\log n)$ for query operation and $O(d\log^2 n)$ for insertion and removal operations, where $n$ and $d$ are the maximal numbers of intervals and pairwise overlapping intervals respectively. We also show that for a monotonic set, that is when no interval properly contains another interval, the amortised complexity is $O(\log n)$ for both query and update operations. We compare the two algorithms for the monotonic interval sets using experiments.
\end{abstract}

\section{Introduction}

Imagine a number of processes all need to use a particular resource for a period of time. Each process $i$ specifies a starting time $s(i)$ and a finishing time $f(i)$ between which it needs to continuously occupy the resource. The resource cannot be shared by two processes at any instance. One is required to design a scheduler which chooses a subset of these processes so that 1)\ there is no time conflict between processes in using the resource; and 2)\ there are as many processes as possible that get chosen.

\smallskip

The above is a typical set-up for the interval scheduling problem, one of the basic problems in the study of algorithms. Formally, given a collection of intervals on the real line specified by their starting and finishing times, the problem asks for a subset of maximal size consisting of pairwise non-overlapping intervals. The interval scheduling problem and its variants appear in a wide range of areas in computer science and applications such as in logistics, telecommunication, and manufacturing. They form an important class of scheduling problems and have been studied under various names and with application-specific constraints~\cite{survey07}.

\smallskip

The interval scheduling problem, as stated above, can be solved by a {\em greedy} scheduler as follows~\cite{kleinberg}. The scheduler sorts intervals based on their finishing time, and then iteratively selects the interval with the least finishing time that is compatible with the intervals that have already been scheduled.
 The set of intervals chosen in this manner is guaranteed to have maximal size. This algorithm works in a {\em static} context in the sense that the set of intervals is given a priori and it is not subject to change.

\smallskip

 In a {\em dynamic} context the instance of the interval scheduling problem is usually changed by a real-time events, and a previously optimal schedule may become not optimal. Examples of such real-time events include job cancelation, arrival of an urgent job, and change in job processing time. To avoid the repetitive work of rerunning the static algorithm every time when the problem instance has changed, there is a demand for efficient {\em dynamic algorithms} for solving the scheduling  problem on the changed instances. In this dynamic context, the set of intervals change through a number of {\em update operations} such as insertion or removal. Our goal  is to design algorithms that allow us to solve the interval scheduling problem in a dynamic setting.
 %efficiently compute a set of pairwise non-overlapping intervals of maximal size after each update operation.

\smallskip

A natural setting for the problem is a  special class of interval sets, which we call {\em monotonic interval sets}. In a monotonic set no interval is properly contained by another interval. For example, if all processes require the same amount of time to be completed, then the set of intervals is monotonic. Moreover, monotonic interval sets are closely related to proper interval graphs. An {\em interval graph} is an undirected graph whose nodes are intervals and two nodes are adjacent if the two corresponding intervals overlap. A {\em proper interval graph} is an interval graph for a monotonic set of intervals. There exist linear time algorithms for representing a proper interval graph by a monotonic set of intervals~\cite{piglinear2004,newpig,piglinear1996}. Furthermore, solving the interval scheduling problem for monotonic intervals corresponds to finding a maximal independent set in a proper interval graph.

\subsection{Related work.} On a somewhat related work, S. Fung, C. Poon and F. Zheng~\cite{fpz} investigated an online version of interval scheduling problem for weighted intervals with equal length (hence, the intervals are monotonic), and designed  randomised algorithms. We also mention that R. Lipton and A. Tompkins~\cite{lt} initiated the study of online version of the interval scheduling problem. In this version a set of intervals are presented to a scheduler in order of start time. Upon seeing each interval the algorithm must decide whether  to include the interval into the schedule.

A related problem on a set of intervals $I$ asks to find a minimal set of points $S$ such that every interval from $I$ contains at least one point from $S$. Such a set $S$ is called a {\em piercing set} of $I$. A dynamic algorithm for maintaining a minimal piercing set $S$ is studied in \cite{pset}. The dynamic algorithm runs in time $O(|S|\log |I|)$. We remark here that if one has a maximal set $J$ of disjoint intervals in $I$, one can use $J$ to find a minimal piercing set of $I$, where each point in the piercing set corresponds to the finishing time of an interval in $J$ in time $O(|J|)$. Therefore
our dynamic algorithm can be adapted to one that maintains a minimal piercing set. Our algorithm improves the results in \cite{pset} when the interval set $I$ is monotonic.

Kaplan et al. in \cite{nested} studied a problem of maintaining a set of nested intervals with priorities. The problem asks for an algorithm that given a point $p$ finds the interval with maximal priority containing $p$. Similarly to our dynamic algorithm, the solution in \cite{nested} also uses dynamic trees to represent a set of intervals.

\subsection{Our results.} For the monotonic case, we provide two dynamic algorithms solving the interval scheduling problem. The first algorithm has $O(\log^2 n)$ amortised complexity for update operations and $O(\log n)$ amortised complexity for the query operations. The second algorithm improves the complexity of update operations to $O(\log n)$ amortised. For the general case, we extend the first algorithm. The complexity of the query operation remains the same, while the complexity of the update operation increases by the factor of $d$, where $d$ is the maximal number of overlapping intervals. Formal explanation  are in the next sections.

The first algorithm maintains  the {\em compatibility forest} data structure denoted  by $\mathsf{CF}$.  We say the {\em right compatible interval} of an interval $i$ is the interval $j$ such that $f(i) < s(j)$ and there does not exist an interval $\ell$ such that $f(i) < s(\ell)$ and $f(\ell) < f(j)$. The $\mathsf{CF}$ data structure maintains the right compatible interval relation. The implementation of the data structure utilises, nontrivially, the dynamic tree data structure of Sleator and Tarjan \cite{tarjan83}. %(Sec.~\ref{sec:pre}).
As a result, in {\bf Theorem~\ref{thm:CF-amortized-time}} and  {\bf Theorem~\ref{thm:CF-gen-amortized-time}}   we prove the amortised bounds for the monotonic and non-monotonic interval sets respectively. % (See Theorem \ref{thm:CF-amortized-time}).

The second dynamic algorithm maintains the {\em linearised tree} data structure denoted  by $\mathsf{LT}$. We say that intervals are {\em equivalent} if their right compatible intervals coincide. The $\mathsf{LT}$ data structure maintains both the right compatibility relation and the equivalence relation.Then, in {\bf Theorem~\ref{thm:LT-amortized-time}} of Section~\ref{sec:LT} we prove that the insertion, removal and query operations take time amortised $O(\log n)$. %See Theorem \ref{thm:LT-amortized-time}).
However, this comes with a cost. As opposed to the $\mathsf{CF}$ data structure that keeps a representation of an optimal set after each update operation, the linearised tree data structure does not explicitly represent the optimal solution.

To test the performance of our algorithms for the monotonic interval sets, we carried out experiments on random sequences of update and query operations.  The experiments show that  the two data structures $\mathsf{CF}$ and $\mathsf{LT}$ perform similarly. The reason for this is that the first dynamic algorithm based on $\mathsf{CF}$  reaches the bound of $\log^{2}n$ only on specific sequences of operations, while on uniformly random sequences the algorithm may run much faster.

\smallskip

\noindent {\em Organisation of the paper.}  Section~\ref{sec:pre} introduces  the problem and monotonic interval sets. Section~\ref{sec:CF} describes the compatibility forest data structure and algorithms for monotonic and non-monotonic interval sets.
Section~\ref{sec:LT} describes the linearsed tree data structures and present our second dynamic algorithm, which is based on the linearised tree. %Section \ref{Extension} extends the data structures by adding the $\mathsf{report}$ operation that outputs the full greedy solution.
Section~\ref{sec:ex} discusses the experiments.
%The last section is a brief conclusion.
%Due to space limitations some proofs are omitted and can be found in the full version of the paper.

\section{Preliminaries}\label{sec:pre}

\noindent{\em Interval scheduling basics.} An {\em interval} is a pair $(s(i),f(i))\in \mathbb{R}^2$ with $s(i) < f(i)$, where $s(i)$ is the {\em starting time} and $f(i)$ is the {\em finishing time} of the interval. We abuse notation and write $i$ for the interval $(s(i), f(i))$. \    Two intervals $i$ and $j$ are {\em compatible} if $f(i) < s(j)$ or $f(j) < s(i)$. Otherwise, these two intervals  {\em overlap}. Given a collection of intervals $I=\{i_1, i_2,\ldots, i_k\}$, a {\em compatible set} of $I$ is a subset $J\subseteq I$ such that the intervals in $J$ are pairwise compatible. An {\em optimal  set} of $I$ is a compatible set of maximal size. The {\em interval scheduling problem} consists of designing an algorithm that  finds an optimal set.

We recall the greedy algorithm that solves the problem \cite{kleinberg}. The algorithm sorts intervals by their finishing time, and then iteratively chooses the interval with the least finishing time  compatible with the last selected interval. The set of thus selected intervals is optimal. The algorithm takes $O(n\log n)$ worst-case time where $n$ is the size of $I$. If the sorting is already given then the algorithm runs in linear time. Below, we formally define the greedy optimal set found by this greedy algorithm.

\smallskip

Let  $\preceq$ \ be the ordering of the intervals by their finishing time. Throughout, by the {\em least interval}, the {\em greatest interval}, the {\em next interval}, the {\em previous interval}, we mean the least, greatest, next and previous interval with respect to  $\preceq$.  Without loss of generality we may assume that the intervals in $I$ have pairwise distinct finishing times. Given the collection $I$, we inductively define the set $J=\{i_1,i_2, \ldots\}$,  the {\em greedy optimal set} of $I$, as follows. The interval $i_1$ is the least interval in $I$. The interval $i_{k+1}$ is the least interval compatible with $i_k$ such that $i_k \prec i_{k+1}$. The set $J$ obtained this way is an optimal set \cite{kleinberg}.

The set $I$ of intervals is called {\em monotonic} if no interval in $I$ contains another interval. The  {\em right compatible interval} of $i$, denoted by $\mathsf{rc}(i)$, is the least interval $j$ compatible with $i$ such that $i \prec j$. Similarly, the {\em left compatible interval} of $i$, written $\mathsf{lc}(i)$, is the greatest interval $j$ compatible with $i$ such that $j \prec i$.

\subsection{Data Structures}
\smallskip
\noindent {\em Binary Search Tree.} A {\em binary search tree} is a standard data structure that maintains a linearly ordered collection of records. The data structure supports the operations
\begin{align*}
&\mathsf{insert}(T,u), \  \mathsf{delete}(T,u), \  \mathsf{find}(T,u), \\
&\mathsf{predecessor}(u), \ \mathsf{successor}(u), \ \mathsf{maximum}(T), \ \text{ and }\ \mathsf{minimum}(T),
\end{align*}
where $T$ is a binary search tree and $u$ is an element from the domain.
If a binary search tree is balanced,  the complexity of all the above operations is $O(\log n)$ where $n$ is the number of elements in the collection. We point out that there are well-known self-balancing binary search tree data structures such as AVL tree and red-black tree.

%Let $I$ be a set of intervals. A natural way to represent the linear ordering $(I,\preceq)$ is to store all intervals in a balanced binary search tree. We call this balanced binary search tree an {\em interval tree} and denote it by $T(I)$. The notion of interval trees used in this paper is a simpler version of the interval trees from \cite{cormen}. An interval tree supports all operations of a binary search tree and performs them in $O(\log n)$ time assuming we use a self-balancing binary search tree implementation.

\medskip

\noindent {\em Splay Tree.}  A {\em splay tree} is also a self-balancing binary search tree for storing linearly ordered objects. In addition to the operations for binary search trees, the splay tree data structure also supports the following operations.
\begin{itemize}
\item $\mathsf{splay}(u)$: This operation reorganises a splay tree so that $u$ becomes the root.
\item $\mathsf{join}(A,B)$: This operation joins two splay trees $A,B$ into one splay tree,  where any interval in $A$ is less than any interval  in $B$, into one tree.
\item $\mathsf{split}(A,u)$: This operation splits the splay tree $A$ of $u$ into two new splay trees
\[
\mbox{$R(u)=\{x\in A\mid u \leq x\}$ \ and \ $L(u)=\{x\in A\mid x < u\}$.}
\]
\end{itemize}
All the operations for splay trees take $O(\log n)$ amortised time~\cite{tarjan85}.

\medskip

\noindent {\em Dynamic Trees.} A dynamic tree data structure maintains a collection of objects that are stored in a number of rooted trees, viewed as directed graphs with edges pointing from children to parents. The trees can be manipulated using the following operations:

\begin{itemize}
\item $\mathsf{link}(v,u)$: If $v$ is the root of a tree and $u$ is a node in another tree, add an edge from $v$ to $u$ and thus ``link'' the trees containing $v$ and $u$ together.
\item $\mathsf{cut}(v)$: If $v$ is not the root of a tree, delete the edge from $v$ to its parent and thus divide the tree containing $v$ into two.
\end{itemize}

These operations have $O(\log n)$ amortised time complexity~\cite{tarjan83}.

\subsection{Problem Setup}
In this setting the collection $I$ of intervals changes over time. Thus, the input to the problem is an arbitrary sequence $o_1, \dots, o_m$ of update and query operations described as follows:

\begin{itemize}
\item  {\em Update operations}:
$\mathsf{insert}(i)$ inserts an interval $i$ and
$\mathsf{remove}(i)$ removes an interval $i$.
\item {\em Query operation}: The operation
$\mathsf{query}(i)$ returns {\sf true} if $i$ belongs to the greedy optimal set and {\sf false} otherwise.
\end{itemize}

Our goal is to design algorithms for performing these operations that minimise the total running time.

\section{Compatibility Forest Data Structure}\label{sec:CF}

In this section we define compatibility forest and describe how to maintain efficiently maintain it for a set of intervals. We first show how to represent a monotonic set, and then we extend the algorithms for the general case.

\subsection{Definition of Compatibility Forest}
Let $I$ be a set of intervals. We define the {\em compatibility forest} as a graph $\mathcal{F}(I)=(V,E)$ where $V=I$ and $(i,j)\in E$ if $j = \mathsf{rc}(i)$. By a forest we mean a directed graph where the edge set contains links from nodes to their parents. We denote the parent of a node $v$ by $p(v)$.  The {\em roots} and {\em leaves} are standard notions that we do not define. Figure~\ref{fig:compat} shows an example of a monotonic set of intervals with its compatibility forest. We note that for every forest one can construct in a linear time a monotonic set of intervals whose compatibility forest coincides (up to isomorphism) with the forest.

\begin{figure}[!ht]
%\vspace{-5mm}
\centering
\begin{tikzpicture}[|-|, scale=0.5, xshift=10, every node/.style={above left}]
		\draw (0.5,3) -- (3,3) node {$a$} ;
		\draw (1,2)   -- (4,2) node {$b$};
		\draw (1.5,1) -- (6,1) node {$c$};
		\draw (2,0)   -- (7.2,0) node {$d$};
		
		\draw (4.5,3) -- (8,3) node {$e$};
		\draw (5.0,2) -- (9,2) node {$f$};
		\draw (6.8,1) -- (11,1) node {$g$};
		
		\draw (10,3) -- (12,3) node {$h$};	
\end{tikzpicture}
\hspace{20mm}		
\begin{tikzpicture}[<-, level distance=7mm, sibling distance = 8mm]
		\node (g) {$g$}
			child {node {$c$}};

		\node (h) [right=15mm] {$h$}
			child {node {$d$}}
			child {node {$e$}
				child {node {$a$}}
				child {node {$b$}}}
			child {node {$f$}};				
\end{tikzpicture}
\caption{\label{fig:compat}Example of a monotonic set of intervals and its compatibility forest.}	
%\vspace{-5mm}
\end{figure}
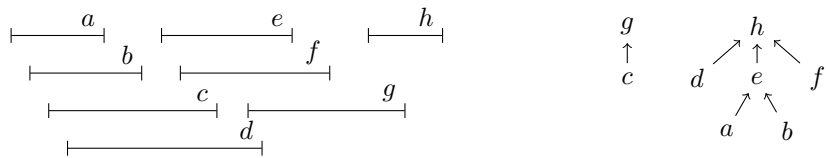

A {\em path} in the compatibility forest $\mathcal{F}(I)$ is a sequence of nodes $i_1,i_2,\ldots,i_k$ where $(i_t,i_{t+1})\in E$ for any $t=1,\ldots,k-1$. It is clear that any path in the forest $\mathcal{F}(I)$ consists of compatible intervals.
Essentially, the forest $\mathcal{F}(I)$ connects nodes by the greedy rule: for any node $i$ in the forest $\mathcal{F}(I)$, if the greedy rule is applied to $i$, then the rule selects the parent $j$ of $i$ in the forest.  Hence, the longest paths in the compatibility forest correspond to an optimal sets of $I$. In particular, the path starting from the least interval is the greedy optimal set. Our first dynamic algorithm amounts to maintaining this path in the forest $\mathcal{F}(I)$.

We explain how we maintain paths in the compatibility forest $\mathcal{F}(I)$. The representation of the forest is developed from the dynamic tree data structure  as in~\cite{tarjan83}.  The idea is to partition the compatibility forest into a set of node-disjoint paths. Paths are defined by two types of edges, {\em solid edges} and {\em dashed edges}.
Each node in the compatibility forest is required to have at most one incoming solid edge.  A sequence of edges $(u_0,u_1),(u_1,u_2),\ldots,(u_{k-1},u_k)$ where each $(u_i,u_{i+1})$ is a solid edge is called a {\em solid path}. A solid path is {\em maximal} if it is not properly contained in any other solid path. Therefore, the solid edges in $\mathcal{F}(I)$ form several maximal solid paths in the forest. Furthermore, the data structure ensures that each node belongs to some maximal solid path. There is an important subroutine in the dynamic tree data structure called the {\em expose} operation~\cite{tarjan83}. The operation starts from a node $v$ and traverses the path from $v$ to the root: while traversing, if the edge $(x, p(x))$ is dashed, we declare $(x,p(x))$ solid and  declare the incoming solid edge (if it exists) incident to $p(x)$ dashed.
Thus, after exposing  node $v$, all the edges on the path from $v$ to the root become solid.  Note that in $\mathsf{CF}$ data structure the $p(x)$ and $\mathsf{rc}(x)$ are the same.

\subsection{Compatibility forest of a monotonic set of intervals}

We denote the representation of $\mathcal{F}(I)$ for the monotonic interval set by $\mathsf{MonCF}$. The representation consists of two components. The first is a binary search tree $T(I)$. The nodes of $T(I)$ are intervals in $I$ ordered by their starting time. Note that monotonicity of $I$ implies that the order of intervals in $T(I)$ coincide with $\preceq$, order of intervals by their finishing time. In addition to standard operations of binary search trees, we define the $\mathsf{right\_compatible}$ operation. Given an interval, $i$ the operation returns $\mathsf{rc}(i)$, if it is in $I$, or $\mathsf{nil}$, otherwise. The second component is a set of splay trees. Each splay tree  stores the nodes of a maximal solid path in the compatibility forest $\mathcal{F}(I)$ %with the underlying order~$\preceq$.
We denote by $\ST_i$ the splay tree containing the interval~$i$.

%We suppose that nodes in the first data structure, as well as in the second data structure point to the same collection of intervals. Hence from a node in a splay tree, one could access the node in the interval tree $T(I)$ corresponding to the same interval, and vice versa, in constant time.
\medskip

\begin{algorithm}[H]
	\caption{$\mathsf{right\_compatible}(i)$ }\label{alg:rc}
	\begin{algorithmic}[1]
        \State $r\leftarrow \mathsf{nil}$
        \State $j\leftarrow$ the root in the interval tree $T(I)$.
        \While{$j\neq \mathsf{nil}$}
            \If{$j\preceq i$ or $j$ overlaps $i$}
                \State $j\leftarrow$ the right child of $j$
            \Else
                \State $r\leftarrow j$
                \State $j\leftarrow$ the left child of $j$
            \EndIf
        \EndWhile
        \State \textbf{return} $r$
	\end{algorithmic}
\end{algorithm}

\begin{lemma}\label{lem:rc}
On monotonic set $I$ of intervals the operation $\mathsf{right\_compatible}(i)$ run in time $\Theta(\log n)$ and return $\mathsf{rc}(i)$.
\end{lemma}
To prove the lemma we observe that for a monotonic set $I$ of intervals and $i,j \in I$, if $i$ overlaps $j$, then each of the intervals between $i$ and $j$ overlaps both $i$ and $j$.

\begin{proof} For the complexity, note that the length of  paths from a leaf to the root in $T(I)$ is $\left \lfloor \log n \right \rfloor + 1$. Thus,  the operation takes time $\Theta(\log n)$.

For the correctness, we use the following loop invariant: \
{\em If $I$ contains $\mathsf{rc}(i)$, then the subtree rooted at $j$ contains $\mathsf{rc}(i)$ or $r$ equals~$\mathsf{rc}(i)$. }

Initially,  $j$ is the root of $T(I)$, so the invariant holds.  Each iteration of the \textbf{while} loop executes either line 5 or lines 7-8 of Alg.~\ref{alg:rc}. If line 5 is executed, then we have $j\preceq i$ or $j$ overlaps $i$. If $j\preceq i$ then all intervals in the left subtree of $j$ are less than $i$. If $j\succeq i$ but $j$ overlaps $i$, then by the observation above,  all intervals between $i$ and $j$ overlap $i$.  In both cases, none of the intervals in the left subtree of $j$ is $\mathsf{rc}(i)$. Therefore setting $j$ to be the right child of $j$ preserves the invariant.

If lines 7-8 are executed, then we have $j\succeq i$ and $j$ is compatible with $i$. If there exists an interval that is less than $j$ and compatible with $i$, then such an interval is in the left subtree of $j$. If such an interval does not exist, $j$ is the smallest interval which is compatible with $i$. Therefore setting $r$ to be $j$ and $j$ to be the right child of $j$ preserves the invariant.

Thus, the algorithm outputs $\mathsf{rc}(i)$ if it exists and outputs $\mathsf{nil}$ otherwise.
Indeed, the loop terminates when $j = \mathsf{nil}$. Hence if the set of intervals $I$ contains $\mathsf{rc}(i)$ then $r = \mathsf{rc}(i)$. If $I$ does not contain $\mathsf{rc}(i)$ then line 5 is executed at every iteration, so $r = \mathsf{nil}$.
\qed
\end{proof}

%%%%%%%%%%%%%%%%%%%%%%%%%%%%%%%%%%%%%%%%%%
We now describe algorithms for maintaining compatibility forest data structure.
We call the algorithms $\mathsf{queryMonCF}$, $\mathsf{insertMonCF}$ and $\mathsf{removeMonCF}$  for the query, insertion, and removal operations, respectively.

\smallskip

\noindent {\em The operation $\mathsf{queryMonCF}$:} To perform this operation on an interval $i$, we first find in the interval tree $T(I)$ the minimum element $m$. We then check if $i$ belongs to the splay tree $\ST_m$.  We return $\mathsf{true}$ if $i \in \ST_m$; otherwise we return $\mathsf{false}$.

\begin{algorithm}[H]
	\caption{$\mathsf{queryMonCF}(i)$}\label{alg:CF-query}
	\begin{algorithmic}[1]
        \State $m \leftarrow \mathsf{minimum}(T(I))$
	\State {\bf return} $\mathsf{find}(\ST_m, i) = i$
    \end{algorithmic}
\end{algorithm}

\noindent {\em The operation $\mathsf{expose}$:} To expose an interval $i$, we find  the maximum element
$j$ in the splay tree $\ST_i$.  Then find the right compatible interval $i'=\mathsf{rc}(j)$.  If $i'$ does not exist (that is,
$j$ is a root in the compatibility forest),  we stop the process. Otherwise, $(j,i')$ is a dashed edge. We split the splay tree at $i'$ into trees $L(i')$ and $R(i')$ and join $\ST_i$ with $R(j')$. We then repeat the process taking $i'$ as~$i$.

\begin{algorithm}[H]
	\caption{$\mathsf{expose}(i)$}\label{alg:expose}
	\begin{algorithmic}[1]
        \State $j\leftarrow \mathsf{maximum}(\ST_i)$
        \State $i'\leftarrow \mathsf{right\_compatible}(j)$
        \While {$i'$ is not $\mathsf{nil}$}
	       \State $\mathsf{split}(\ST_{i'}, i')$
           \State $\mathsf{join}(\ST_i, R(i'))$
           \State $j\leftarrow \mathsf{maximum}(\ST_{i'})$
           \State $i'\leftarrow \mathsf{right\_compatible}(j)$
        \EndWhile
    \end{algorithmic}
\end{algorithm}

\noindent {\em The operation $\mathsf{insertMonCF}$:} To insert an interval $i$, we add $i$ into the tree $T(I)$. Then we locate the next interval $r$ of $i$ in the ordering $\preceq$. If such $r$ exists, we access  $r$ in the splay tree $\ST_r$ and find the interval $j$ such that $(j,r)$ is a solid edge. If such a $j$ exists and $j$ is compatible with $i$, we delete the edge $(j,r)$ and create a new edge $(j,i)$ and declare it solid.  We restore the longest path of the compatibility forest by exposing the least interval in $T(I)$.

\begin{algorithm}[H]
	\caption{$\mathsf{insertMonCF}(i)$}\label{alg:CF-insert}
	\begin{algorithmic}[1]
        	\State $\mathsf{insert}(T(I), i)$  			%\Comment Insert $i$ into the interval tree $T(I)$
	       	\State $r \leftarrow \mathsf{next}(i)$ \Comment Find the next interval of $i$
		\If {$r \neq \mathsf{nil}$}
	       		\State $j \leftarrow \mathsf{predecessor}(\ST_r, r)$ \Comment Find a solid edge $(j,r)$
			\If {$j \neq \mathsf{nil}$ and $j$ is compatible with $i$}       	  		
				\State $\mathsf{split(\ST_{r}, r)}$  \Comment Destroy the solid edge $(j,r)$
				%\State $\mathsf{insert(L(r), i)}$   \Comment Insert $i$ into the splay tree $L(r)$
			\EndIf
		\EndIf
		\State $\mathsf{expose}(\mathsf{mininum}(T(I)))$
	\end{algorithmic}
\end{algorithm}
\smallskip

\noindent {\em The operation $\mathsf{removeMonCF}$:} To delete an interval $i$, we delete the incoming and outgoing solid edges of $i$ if such edges exist. We then delete $i$ from the tree $T(I)$. We restore the longest path of the $\mathsf{CF}$ by exposing the least interval in $T(I)$. %See Alg.~6 in the Appendix.

\begin{algorithm}[H]
	\caption{$\mathsf{removeMonCF}(i)$}\label{alg:CF-delete}
	\begin{algorithmic}[1]
       	\State $\mathsf{remove}(\ST_i, i)$ \Comment Delete $i$ from its splay tree $\ST_i$
        	\State $\mathsf{remove}(T(I), i)$ \Comment Delete $i$ from the interval tree $T(I)$
		\State $\mathsf{expose}(\mathsf{minimum}(T(I)))$
	\end{algorithmic}
\end{algorithm}

\noindent {\em Correctness.} For correctness of operations, we use the following invariants.
\begin{itemize}
\item[(A1)] Every splay tree represents a maximal path formed from solid edges.
%For all intervals $i$ and $j$, if $j$ is the previous interval of $i$ in the splay tree $\ST_i$, then $(j,i)$ is an edge in the forest $\mathcal{F}(I)$.
\item[(A2)] Let $m$ be the least interval in $I$. The splay tree $\ST_m$ contains all intervals on the path from $m$ to the root.
\end{itemize}
Note that (A2) guarantees that the query operation correctly determines if a given interval $i$ is in the greedy optimal set. \ The next lemma shows that (A1) and (A2) are invariants indeed and that the operations correctly solve the dynamic monotonic interval scheduling problem.

\begin{lemma}\label{lem:CF-correct}
(A1) and (A2) are invariants of $\mathsf{insertMonCF}$, $\mathsf{removeMonCF}$, and $\mathsf{queryMonCF}$.
% operations.

\end{lemma}

\begin{proof}
For (A1), first consider the operation of joining  two splay trees $A$ and $B$ via the operation $\mathsf{expose}(i)$.
Let $j$ be the maximal element in $A$ and $j'$ be the minimum element in $B$. In this case, $j'$ is obtained by the operation $\mathsf{right\_compatible}(j)$.
%By Lemma~\ref{lem:rc2}
It is clear that $(j,j')$ is an edge in the forest $\mathcal{F}(I)$. \ Next, consider the case when we apply $\mathsf{insertMonCF}(i)$  into the splay tree $A$. In this case,  $A$ is $L(r)$ where $r$ is the next interval of $i$ in $I$. Let $j$ be the previous interval of $r$ in the tree $\ST_r$. By (A1), before inserting $i$, $(j,r)$ is an edge in $\mathcal{F}(I)$ and thus $r=\mathsf{rc}(j)$. Note we only insert $i$ to $L(r)$ when $j$ is compatible with $i$. Since $i<r$, after inserting $i$, $i$ becomes the new right compatible interval of $j$. So, joining $L(r)$ with $i$ preserves (A1).
Operations $\mathsf{removeMonCF}(i)$ and  $\mathsf{queryMonCF}(i)$ do not create new edges in splay trees.  Thus,
(A1) is preserved under all operations.

For (A2), the $\mathsf{expose}(i)$ operation terminates when it reaches a root of the compatibility forest. As a result, $\ST_i$ contains all nodes on the path from $i$ to the root.
% the $\mathsf{right_compatible}$ operation returns $\mathsf{nil}$. This implies that $j$ is a root operations
%and when it terminates the splay tree $\ST_i$ contains all nodes on the path from $i$ to the root.
Since $\mathsf{expose}(\mathsf{minimum}(T(I)))$ is called at the end of both $\mathsf{insertMonCF}(i)$ and $\mathsf{removeMonCF}(i)$ operations, (A2) is preserved under every operation.
\qed
\end{proof}

\noindent {\em Complexity.}  Let $n$ be the number of intervals in $I$. As discussed in Section~\ref{sec:pre}, all operations for the interval tree have $O(\log n)$ worst case complexity, and all operations for splay trees have $O(\log n)$ amortised complexity. The query operation, involves finding the minimum interval in $T(I)$ and searching $i$ in a  splay tree. Hence, the query  operation runs in amortised time $O(\log n)$.
For each insert and remove operation, we perform a constant number of operations on $T(I)$ and the splay trees plus one $\mathsf{expose}$ operation.

%The analysis of $\mathsf{expose}$ operation is similar to the analysis of dynamic trees in \cite{tarjan83}.
To analyse $\mathsf{expose}$ operation, define the size $\mathsf{size}(i)$ of an interval $i$ to be the number of nodes in the subtree rooted at $i$ in $\mathcal{F}(I)$. Call an edge $(i,j)$ in $\mathcal{F}(I)$  {\em heavy} if $2\cdot \mathsf{size}(i)>\mathsf{size}(j)$, and {\em light} otherwise. It is not hard to see that this partition of edges has the following properties:
\begin{itemize}
\item[($\star$)] Every node has at most one incoming heavy edge.
\item[($\star\star$)] Every path in the compatibility forest consists of at most $\log n$ light edges.
\end{itemize}

\begin{lemma}\label{lem:CF-complex}
In a sequence of $k$ update operations, the total number of dashed edges, traversed by $\mathsf{expose}$ operation,  is $O(k\log n)$.
\end{lemma}
\begin{proof}
The number of iterations in $\mathsf{expose}$ operation is the number of dashed edges in a path from the least interval to the root. A dashed edge is either heavy or light. From ($\star\star$), there are at most  $\log n$ light dashed edges in the path. To count the number of heavy dashed edges, consider the  previous update operations.  After deletion of
$i$, all children of $i$ become children of the next interval of~$i$. After inserting $i$, the children of the next interval of $i$ that are compatible with $i$ become children of~$i$.
Therefore , there are at most two path where an update operation transforms light dashed edges to heavy dashed edges. Figure~\ref{fig:update} illustrates these structural changes. Since there are at most $\log n$ light dashed edges on each path, an update operations creates at most $\log n$ heavy dashed edges.

Execution of $\mathsf{expose}$ in an update operation creates at most $\log n$ heavy dashed edges from heavy solid edges. Hence, the total number of heavy dashed edges created after $k$ update operations is $O(k\log n)$. \qed
\begin{figure}[!ht]
	%trim=l b r t
	%\vspace{-6mm}
    \centering
	\includegraphics[clip=true, trim=0 0 0 0, width=1\textwidth]{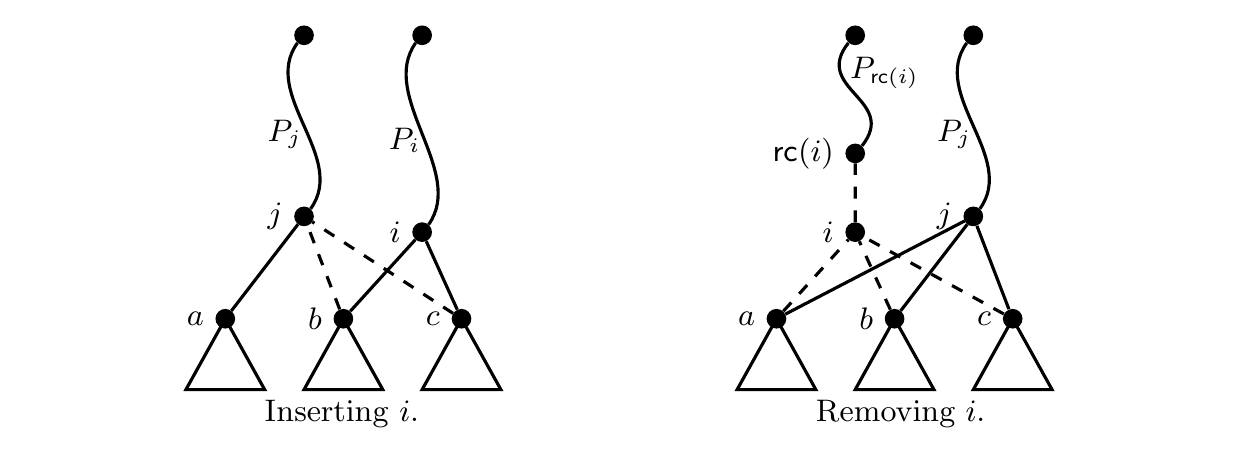}
	\caption{Redirections of edges in $\mathsf{CF}$, where $j$ is the next interval of $i$.}
	\label{fig:update}
	%\vspace{-6mm}
\end{figure}
\end{proof}
Lemma~\ref{lem:CF-correct} and Lemma~\ref{lem:CF-complex} give us the following theorem:

\begin{theorem}\label{thm:CF-amortized-time}
The algorithms $\mathsf{queryMonCF}$, $\mathsf{insertMonCF}$ and $\mathsf{removeMonCF}$ solve the dynamic monotonic interval scheduling problem. The algorithms perform insert interval and remove interval operations in $O(\log^2 n)$ amortised time and query operation in $O(\log n)$ amortised time, where $n$ is the size of the set $I$ of intervals.
\end{theorem}

\noindent {\bf Remark}. {Tarjan and Sleator's dynamic tree data structure has amortised time $O(\log n)$ for update and query operations. To achieve this, the algorithm maintains dashed edges explicitly. Their technique cannot be adapted directly to $\mathsf{CF}$ because insertion or removal of  intervals may result in redirections of a linear number of edges. An example is depicted on the Figure~\ref{fig:update}.% in Appendix~\ref{apx:CF}.
Therefore, more care should be taken;  \ for instance, one needs to maintain dashed edges implicitly in  $T(I)$ and compute them calling $\mathsf{right\_compatible}$ operation.}

%Lemmas~\ref{lem:CF-complex} and~\ref{lem:rc} provide us the result stated in the theorem above.

\smallskip

\begin{proposition}[{Sharpness of the $\log^2n$ bound}]
In $\mathsf{CF}$ data structure there exists a sequence of $k$ update operations with $\Theta(k\log^2 n)$ total running time.
\end{proposition}

\begin{proof}
Consider a sequence which creates a set of $n < k$ intervals. We assume that $n = 2^{h+1}-1$, where $h \in \N$. The first $n$ operations of the sequence are $\mathsf{insertMonCF}$  such that the resulted compatibility forest is a perfect binary tree $T_n$, that is, each internal node of $T_n$ has exactly two children and the height of each leaf  in $T_n$ is $h$.  The next $k-n$ operations starting form $T_n$ are pairs of $\mathsf{insertMonCF}$  followed by $\mathsf{removeMonCF}$.
At stage $s=n+2m+1$, $\mathsf{insertMonCF}$ inserts an interval $i_s$ into $T_{s}$ producing the tree $T_{s+1}$. The interval $i_s$ is such that in  $T_{s+1}$ the path from $i_s$ to the root is of length $h+1$ and the path consists of dashed edges only. Then, at stage $s+1$ we delete $i_s$. This produces a tree $T_{s+2}$ which is a perfect binary tree of height $h$. We repeat this $k-n$ times. We can select $i_s$ as desired since
each  perfect binary tree $T_s$ always has a path of length $h$ consisting of dashed edges only. Therefore a sequence of $k$ such operations takes time~$\Theta(k\log^2 n)$.\qed
\end{proof}

\subsection{Compatibility forest of a non-monotonic set of intervals}

In this section, we show how to maintain compatibility forest for a set of non-monotonic intervals.  In this case, extra care should be taken when we insert interval $i$ that is covered by other intervals since $i$ may become the new right compatible interval for several overlapping intervals.  The example in Figure~\ref{fig:bad-insert} shows such insertion. Therefore, when we insert an interval $i$, we need to  find all intervals covering $i$.

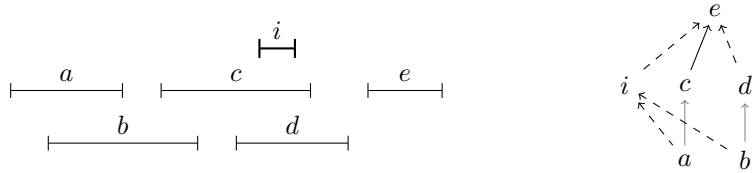
\begin{figure}[h!]
\centering
\begin{tikzpicture}[xscale=0.5, yscale=0.7]
	
	\I{a}{0,1}{3};
	\I{b}{1,0}{4};

	\I{c}{4,1}{4};
	\I{d}{6,0}{3};

	\I{e}{9.5,1}{2};

	\begin{scope}[thick]
		\I{i}{6.6,1.8}{1};
	\end{scope}

	\node at (0,-0.5) {};	
\end{tikzpicture}
\hspace{20mm}
\begin{tikzpicture}[<-, xscale=0.4, yscale=0.5, i'/.style={color=blue}]
  \foreach \name/\x/\y in {a/0/0, b/2/0,
  	i/-2/2, c/0/2, d/2/2, e/1/4 }
	\node (\name) at (\x, \y) {$\name$};
	
  \path (e) edge[dashed] (i) edge[dashed] (d) edge (c);
  \path (c) edge[gray, very thin] (a);	
  \path (d) edge[gray, very thin] (b);	
  \path (i) edge[dashed] (b) edge[dashed] (a);	

\end{tikzpicture}
\caption{Insertion of $i$ into $\mathsf{CF}$ destroys solid edges $(a,c)$ and $(b,d)$}\label{fig:bad-insert}
\end{figure}

We describe a new operation $\mathsf{covers}(i)$, which returns all intervals covering a given interval. To support this operation, we introduce the third component - an interval tree data structure. An {\em interval tree} \cite{i_tree}  is a leaf-oriented balanced binary search tree where leaves store endpoints of the intervals in increasing order. Intervals themselves are stored in the internal nodes as follows. For each internal node $v$ the set $I(v)$ consists of intervals that contain the {\em split point} of $v$ and are covered by the {\em range} of $v$. The split point of $v$, denoted by $split(v)$, is a number such that the leaves of the left subtree of $v$ store endpoints smaller than $split(v)$, and the leaves of the right subtree of $v$ store endpoints greater than $split(v)$. The range of $v$, denoted by $range(v)$, is defined recursively as follows. The range of the root is $(-\infty, \infty]$. For a node $v$, where $range(v) = (l, r]$, the range of the left child of $v$ is $(l,split(v)]$, and the range of the right child of $v$ is $(split(v),r]$.

To allow insertions and deletions of intervals in the interval tree, we represent it as a red-black tree $IT$. In a red-black tree, insertion or deletion of a node takes $O(\log n)$ time plus the time for at most 3 rotations to restore the balance. When performing a rotation around an edge $(v, p(v))$ the sets $I(v)$ and $I(p(v))$ change. Let the range of $p(v)$ be $(\ell, r]$. If $v$ is the left child, the range of $p(v)$ after rotation becomes $[split(v), r]$. If $v$ is the right child, the range of $p$ shortens at the other end and becomes $[\ell, split(v)]$. Therefore all intervals in $I(p(v))$ that intersects with $split(v)$ must be moved to $I(v)$. Note that ranges of other nodes are not affected. We represent $I(v)$ with two binary search trees. The first tree $T_s(v)$ stores intervals of $I(v)$ sorted from left to right by their starting point. The second tree $T_f(v)$ stores intervals of $I(v)$ sorted from right to left by their finishing point. To move intervals from $I(v)$ to $I(u)$, we perform join and split operations on the trees representing these sets. Thus in total we need $O(\log n)$ time to insert or delete a node in $IT$.

To find all intervals covering interval $i$ we do the following. We walk down in the interval tree starting at the root. At every node $v$, we compare $split(v)$ with $s(i)$ and $f(i)$. If $split(v) > f(i)$ or $split(v) < s(i)$, we respectively traverse the tree $T_s(v)$ or $T_f(v)$ from left to right and report all intervals that covers $i$. We continue to the left or right child of $v$ respectively. Otherwise, the split point $split(v)$ intersect with $i$. We traverse $T_s(v)$, report intervals covering $i$. We terminate the search at this node, because $i$ intersects with ranges of both children of $v$, namely $[\ell, split(v))$ and $[split(v), r)$, and therefore cannot be covered by any interval, fully contained in these ranges. Formally, the operation is described in Algorithm~\ref{alg:covered}.

\begin{algorithm}[H]
	\caption{$\mathsf{covers}(i)$}\label{alg:covered}
	\begin{algorithmic}[1]
	\State $S \leftarrow \mathsf{nil}$% \Comment Set of intervals covering $i$
        \State $v \leftarrow$ root of $IT$
        \While{$v \ne \mathsf{nil}$ }
        		\If{$split(v) < s(i)$}
			\State $x \leftarrow \mathsf{minimum}(T_f(v))$
			\State $v \leftarrow$ right child of $v$	
		\Else
			\State $x \leftarrow \mathsf{minimum}(T_s(v))$
			\State $v \leftarrow$ left child of $v$	
		\EndIf
		\While{$x$ covers $i$}
			\State Add $i$ into $S$
			\State $x \leftarrow \mathsf{successor}(x)$% \Comment None of the successors of $x$ covers $i$
		\EndWhile
		\If{$split(v) \in i$ }
				\State {\bf break}
		\EndIf
	\EndWhile
	\State {\bf return} $S$
    \end{algorithmic}
\end{algorithm}

Another thing we need to take care of is the $\mathsf{right\_compatible}$ operation. Since the set of intervals is not monotonic, the observation, essential for the proof of Lemma~\ref{lem:rc}, does not hold. Namely, there might exists intervals $i \preceq  j$ and $k$ covered by $j$ such that $i$ intersects with $j$, but $i$ is compatible with $k$. To overcome this difficulty, we augment the search tree $T(I)$ and the operation $\mathsf{right\_compatible}$ as follows. In every node $v$ of $T(I)$ we keep a pointer to the interval $m$ in the subtree rooted at $v$ with the smallest finishing time. It is not hard to check that these pointers can be updated after a rotation in constant time. Therefore, maintenance of these pointers does not change the asymptotic complexity of operations on $T(I)$.

Recall that when we search for the right compatible interval of $i$ in the monotonic set, we go to the left child of the current interval $j$ if $i \preceq j$ and $i$ is compatible with $j$. In a non-monotonic set, we need to check  if there is an interval $k$ that is covered by $j$ and hence compatible with $i$. If such $k$ exists, it is in the right subtree of $j$ and we can access it in constant time using the pointers we described above. Therefore, if we go to the left child of $x$, we remember an interval with the smallest finishing time among three intervals: the last remembered interval, an interval at $x$ or an interval with the smallest finishing time in the right subtree of~$x$. For example, if we search for $\mathsf{rc}(a)$ in the tree shown in Figure~\ref{fig:rc_gen_example}, we traverse the path $\{g, d, b, c\}$ and remember intervals $h, d, d$ one after another.

\begin{figure}[!ht]
\centering
\begin{tikzpicture}[|-|, xscale=0.40, yscale=0.6, every node/.style={above left}]
		\draw (2, 4) -- (4.5, 4) node {$b$} ;
		\draw (5, 4) -- (8.2, 4) node {$d$} ;
		\draw (8.5, 4) -- (17, 4) node {$g$} ;
		
		\draw (1, 3)   -- (3, 3) node {$a$};
		\draw (4, 3)   -- (10.5, 3) node {$c$};
		\draw (11.5, 3)   -- (16, 3) node {$i$};
		
		\draw (6, 1) -- (9.5, 1) node {$e$};
		\draw (11, 2) -- (14, 2) node {$h$};
		\draw (14.5, 2) -- (18, 2) node {$k$};
		
		\draw (6.5, 2) -- (10, 2) node {$f$};
		\draw (13.5, 1) -- (15, 1) node {$j$};
		
		\node at  (2,0) {};
\end{tikzpicture}
\hspace{2mm}
\begin{tikzpicture}[-, scale=0.85,
	level distance=12mm, sibling distance = 32mm,
	level 2/.style={level distance=8mm, sibling distance=15mm},
	level 3/.style={level distance=8mm, sibling distance=7mm} ]	
  \node (g) at (0,0) [yshift=10cm] {$g$}
    child { node {$d$}
      child { node {$b$}
      	child {node {$a$} }
	child {node {$c$} }
      }
      child { node {$e$}
        child [missing]
        child { node {$f$} } }
      }
    child { node {$j$}
      child { node {$h$}
        child [missing]
        child { node {$i$} } }
      child { node {$k$} }
      };		
\end{tikzpicture}
\caption{Figure}\label{fig:rc_gen_example}
\end{figure}

\medskip

We are now ready to describe the operations, that maintain a compatibility forest for a non-monotonic set of intervals. The query operation $\mathsf{queryCF}(i)$ and the remove operation $\mathsf{removeCF}(i)$ are identical to $\mathsf{queryMonCF}(i)$ and \linebreak $\mathsf{removeMonCF}(i)$ respectively. The insert operation $\mathsf{insertCF}(i)$ does the following. First, we add $i$ into the trees $T(I)$ and $IT$. Second, as in the monotonic case, we check is there exists a solid edge $(j, r)$ such that $i$ substitutes $r$. Namely, we search for $r$ such that: (i) $i \prec r$, (ii) $i$ is not covered by $r$, (iii) for every $i \prec \ell \prec r$, $\ell$ covers $i$. Then, if there exist a solid edge from $j$ to $r$ and $j$ is compatible to $i$, we make this edge dashed. Third, for every interval, which covers $i$, we make the incoming solid edge, if any, dashed. Finally, we restore the longest path of the compatibility forest by exposing the $\preceq$-least interval. Algorithm~\ref{alg:CF-insert-gen} describes the operation in details.

\begin{algorithm}[H]
\caption{$\mathsf{insertCF}(i)$}\label{alg:CF-insert-gen}
\begin{algorithmic}[1]
        	\State $\mathsf{insert}(T(I), i)$  			%\Comment Insert $i$ into the interval tree $T(I)$
	\State $\mathsf{insert}(IT, i)$

	\State $r \leftarrow\, \preceq$-next interval that does not cover $i$%   \mathsf{next}(i)$ \Comment Find the next interval of $i$
	\If {$r \neq \mathsf{nil}$}
		\State $j \leftarrow \mathsf{predecessor}(\ST_r, r)$ \Comment Find a solid edge $(j,r)$
		\If {$j \neq \mathsf{nil}$ and $j$ is compatible with $i$}       	  		
			\State $\mathsf{split(\ST_{r}, r)}$
		\EndIf
	\EndIf

	\For {$c$ in  $\mathsf{covers}(i)$}
       		\State $j \leftarrow \mathsf{predecessor}(\ST_c, c)$ %\Comment Find a solid edge $(j,r)$
		\If {$j \neq \mathsf{nil}$ and $j$ is compatible with $i$}       	  		
			\State $\mathsf{split(\ST_{c}, c)}$
		\EndIf
	\EndFor
	\State $m \leftarrow$ $\preceq$-least interval
	\State $\mathsf{expose}(m)$
\end{algorithmic}
\end{algorithm}

\begin{lemma}
In a sequence of $k$ update operations, the total number of dashed edges, traversed by $\mathsf{expose}$ operation,  is $O(d\cdot k\log n)$, where $d$ is the size of a maximal subset of pairwise overlapping intervals.
\end{lemma}
\begin{proof}
The proof is similar to the proof of Lemma~\ref{lem:CF-complex}. Recall, that an edge $(i,j)$ in the compatibility forest is heavy if the number of nodes in the tree rooted at $i$ is two times greater than the number of nodes in a tree rooted at $j$. We count the number of heavy dashed edges created by the sequence of $k$ update operations.

Let $I_t$ be the interval after performing $t$ operations. We define $d$ as follows
$$d = \underset{1 \leq t \leq k}{\mathsf{max}}\{|J| \mid J \subseteq I_t \text{ and for every } i,j \in J,\, i \text{ overlap } j \}.$$

Let $i$ be deleted or inserted interval. In the monotonic case, there is at most one interval that exchanges children with $i$. Here we need to take into account all the intervals $j_1, \dots, j_m$, that covers $i$. When we insert $i$, all children of each $j_t$ becomes children of $i$. When we delete  $i$, the children of $i$ are distributed among the intervals, covering $i$. Since there are at most $d$ such intervals, and every path from $j_t$ to the root has at most $\log n$ light dashed edges, an update operation creates at most $d\cdot \log n$ heavy dashed edges. Thus, in total  $\mathsf{expose}$ operation traverses $O(d\cdot k\log n)$ dashed edges in a sequence of $k$ operations. \qed
\end{proof}

\begin{theorem}\label{thm:CF-gen-amortized-time}
The algorithms $\mathsf{queryCF}$, $\mathsf{insertCF}$ and $\mathsf{removeCF}$ solve the dynamic interval scheduling problem. The algorithms perform insert interval and remove interval operations in $O(d\log^2 n)$ amortised time and query operation in $O(\log n)$ amortised time, where $n$ is the size of the set $I$ of intervals and $d$ is the size of a maximal subset of pairwise overlapping intervals.
\end{theorem}

%%%%%%%%%%%%%%%%%%%%%%%%%%%%%%%%%%%%%%%%%%
\section{Linearised Tree  Data Structure}\label{sec:LT}

In this section, we develop a new data structure for the dynamic interval scheduling problem. The dynamic algorithm based on this data structure performs all operations in amortised $O(\log n)$ time. However, the algorithm requires the interval set to be monotonic at all times.

\subsection{Definition of Linearised Tree}

We say that intervals $i$ and $j$ are {\em equivalent}, written as $i\sim j$,  iff $\mathsf{rc}(i)=\mathsf{rc}(j)$. Denote the equivalence class of  $i$ by $[i]$. Thus, two intervals are in the same equivalence class  if they are siblings in the compatibility forest. In the linearised tree we arrange all intervals in an equivalence class in a path using the $\preceq$-order. The linearised tree consists of all such ``linearised'' equivalence classes joined by edges. Hence, there are two types of edges in the linearised tree. The first type connects intervals in the same equivalence class. The second type joins the greatest interval in an equivalence class with its right compatible interval. Formally, the {\em linearised tree} ${\cal L}(I)$ is a triple $(I;E_\sim, E_c)$, where $E_\sim$ and $E_c$ are disjoint set of edges such that:

\begin{itemize}
\item $(i,j)\in E_{\sim}$ if and only if $i\sim j$ and $i$ is the previous interval of $j$.
Call $i$ the {\em equivalent child} of $j$.
\item $(i, j) \in E_c$ if and only if $i$ is  the greatest interval in $[i]$ and $j=\mathsf{rc}(i)$. Call $i$ the {\em compatible child} of $j$.
\end{itemize}

Figure~\ref{fig:LT1} shows an example of a linearised tree. We stress three crucial differences between the $\mathsf{CF}$ and $\mathsf{LT}$ data structures. The first is that a path in a linearised tree may not be a compatible set of intervals. The second is that  linearised trees are binary. The third is when we insert or remove an interval we need to redirect at most two existing edges in the linearised tree. We explain the last fact in more details below when we introduce the dynamic algorithm.

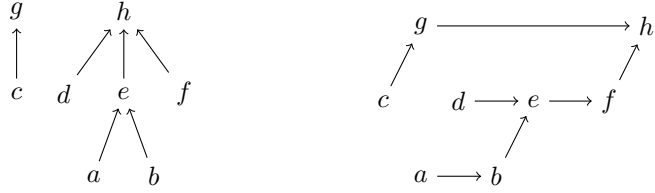
\begin{figure}[!ht]
%\vspace{-5mm}
\centering	
\begin{tikzpicture}[<-, level distance=11mm, sibling distance = 8mm]
		\node (g) {$g$}
			child {node {$c$}};

		\node (h) [right=12mm] {$h$}
			child {node {$d$}}
			child {node {$e$}
				child {node {$a$}}
				child {node {$b$}}}
			child {node {$f$}};				
\end{tikzpicture}
\hspace{20mm}
\begin{tikzpicture}[<-, scale = 0.5]
	\foreach \name/\x/\y in {a/0/0, b/2/0, d/1/2, e/3/2, f/5/2, c/-1/2,  g/0/4, h/6/4}
		\node (\name) at (\x, \y) {$\name$};

	\path (b) edge (a);
	\path (g) edge (c);
	\path (f) edge (e);
	\path (e) edge (d) edge (b);
	\path (h) edge (g) edge (f);
	
\end{tikzpicture}
\caption{\label{fig:LT1} Example of a compatibility forest (left) and  linearised tree (right).}	
%\vspace{-5mm}
\end{figure}

We use the dynamic tree data structure to represent the linearised tree. %Every edge of the tree is maintained explicitly in the same data structure.
We also maintain the interval tree $T(I)$
%described in Section~\ref{sec:pre}
as an auxiliary data structure. The interval tree is used to compute previous and next intervals as well as left compatible and right compatible intervals of a given interval.

\medskip

\subsection{Maintaining Linearised Tree}
%%%%%%%%%%%%%%%%%%%%%%%%%%%%%%%%%%%%%%%%%%
To maintain $\mathsf{LT}$ we will need $\mathsf{left\_compatible}(i)$ operations, which returns $\mathsf{lc}(i)$ if it is in the interval set or $\mathsf{nil}$ otherwise. Algorithms~\ref{alg:lc} defines this operation.

\begin{algorithm}[H]
	\caption{$\mathsf{left\_compatible}(i)$}\label{alg:lc}
	\begin{algorithmic}[1]
        \State $\ell\leftarrow \mathsf{nil}$
        \State $j\leftarrow$ the root in the interval tree $T(I)$.
        \While{$j\neq \mathsf{nil}$}
            \If{$j\succeq i$ or $j$ overlaps $i$}
                \State $j\leftarrow$ the left child of $j$
            \Else
                \State $\ell\leftarrow j$
                \State $j\leftarrow$ the right child of $j$
            \EndIf
        \EndWhile
        \State \textbf{return} $\ell$
	\end{algorithmic}
\end{algorithm}

We now describe algorithms for maintaining linearised tree data structure. We call the algorithms $\mathsf{queryLT}$, $\mathsf{insertLT}$ and $\mathsf{removeLT}$  for the query, insertion, and removal operations, respectively.
\smallskip

\noindent {\em The operation $\mathsf{queryLT}$:} To detect if an interval $i$ is in the greedy optimal set,
consider the path  $P$  from the least node $m$ to the root in the linearised tree~${\cal L}(I)$. If $i \notin P$, return $\mathsf{false}$. Otherwise, consider the direct predecessor $j$ of $i$ in the path $P$. If $j$ does not exist or $(j,i) \in E_c$, return $\mathsf{true}$. Otherwise, we return $\mathsf{false}$.

\begin{algorithm}[H]
	\caption{$\mathsf{queryLT}(i)$}\label{alg:LT-query}
	\begin{algorithmic}[1]
        \State $m \leftarrow \mathsf{minimum}(T(I))$
 				\If{$i = m$} 																	\Comment $i$ is the least interval
        		\State {\bf return} $\mathsf{true}$						
	      \EndIf
        \State $\mathsf{expose}(m)$										\Comment Make the path from $m$ to the root solid
 				\If{$i \neq \mathsf{find}(\ST_m, i)$}          \Comment $i$ is not on the path from $m$ to the root
                \State {\bf return} $\mathsf{false}$
        \EndIf
				
        \State $j \leftarrow \mathsf{predecessor}(\ST_{m}, i)$ \Comment $(j,i)$ is an edge in $\mathsf{LT}$
         \If{$i$ is compatible with $j$}               %\Comment $i$ is the right compatible interval of $j$
        		\State  {\bf return} $\mathsf{true}$						
	    \Else
	    	\State  {\bf return} $\mathsf{false}$
       \EndIf
	\end{algorithmic}
\end{algorithm}

\begin{lemma}\label{lem:query-LT}
The operation $\mathsf{queryLT}(i)$ returns $\mathsf{true}$ if and only if a given interval $i$ belongs to the greedy optimal set of $I$.
\end{lemma}

\begin{proof}
Let $J$ be the greedy optimal set of $I$ and $m$ be the least element of $I$. Suppose the algorithm $\mathsf{queryLT}(i)$ outputs $\mathsf{true}$. This can happen when (1) $i = m$.
%$\mathsf{predecessor}(\ST_{m},i)=\mathsf{nil}$.
In this case $i$ is the least element of $I$, hence $i$ belongs to $J$; or (2) $i$ is compatible with $\mathsf{predecessor}(\ST_{m},i)$. Note that for every interval $x$ from $\ST_{m}$ there exists an interval $y$ from the greedy set $J$ such that $y\in[x]$. Consider such an interval $y\in J\cap[\mathsf{predecessor}(\ST_{m},i)]$. Since $y\in J$ and $i$ is the next compatible interval of $y$, $i$ belongs to the greedy optimal set $J$.

It is not hard to see by induction on the number of elements in $J$ that $J\subseteq \ST_{m}$. Suppose the algorithm $\mathsf{queryLT}(i)$ outputs
$\mathsf{false}$. It happens in two cases. First, $i$ is not in $\ST_{m}$. Then $i\notin J$. Second, $\mathsf{predecessor}(\ST_{m},i)$ exists and is not compatible with $i$. Then $i$ is not the least interval in $[i]\cap \ST_{m}$, but every element $x\in J$ is the least element in $[x]\cap \ST_{m}$. Hence $i\notin J$.
\qed
\end{proof}

\noindent {\em The operation $\mathsf{insertLT}$:} Given $i$, we insert $i$ into $T(I)$.
If $i$ is the greatest interval in $[i]$, then we add the edge $(i, \mathsf{rc}(i))$ into $E_c$. Otherwise,
we add the edge $(i,j)$ to $E_{\sim}$, where $j$ is the next interval equivalent to $i$.
%We also add at most two incoming edges to $i$.
If $i$ has an equivalent child $k$ then we add the edge $(k,i)$ to $E_{\sim}$ and delete the old outgoing edge from $k$ in case such edge exists.  If $i$ has a compatible child $\ell$ then we add the edge $(\ell, i)$ to $E_c$ and delete the old outgoing edge in case such edge exists.

\begin{algorithm}[H]
	\caption{$\mathsf{insertLT}(i)$}\label{alg:LT-insert}
	\begin{algorithmic}[1]
        \State $\mathsf{insert}(T(I),i)$ 								%\Comment Insert $i$ to the interval tree $T(I)$.
	\If{$i$ is not the greatest interval in $I \cup \{i\}$}   \Comment $i$ has a parent          		
	    \If {$\mathsf{lc}(\mathsf{rc}(i)) = i$}								\Comment $(i, \mathsf{rc}(i)) \in E_c$
                \State $\mathsf{link}(i,\mathsf{rc}(i))$      %\Comment Add an edge to $\mathsf{rc}(i)$
	    \Else																									\Comment $(i, \mathsf{next}(i)) \in E_\sim$
	       \State $\mathsf{link}(i, \mathsf{next}(i))$          %\Comment Add an edge to the next interval of $i$
           \EndIf
        \EndIf
        \State $j\leftarrow \mathsf{previous}(i)$							
         \If{$\mathsf{rc}(j)=\mathsf{rc}(i)$}									\Comment $(j, i) \in E_\sim$
	       \State $\mathsf{cut}(j)$ and $\mathsf{link}(j,i)$
        \EndIf
         \State $j \leftarrow \mathsf{lc}(i)$
        \If{$\mathsf{rc}(j) = i$}															\Comment $(j, i) \in E_c$
            \State $\mathsf{cut}(j)$ and $\mathsf{link}(j,i)$
        \EndIf
	\end{algorithmic}
\end{algorithm}

\noindent {\em The operation $\mathsf{removeLT}$:}  Given $i$, we delete $i$ from $T(I)$. We delete an edge from $i$ to the parent of $i$ and redirect the edge from the equivalent child $j$ of $i$ to the parent of $i$. Then we redirect an edge from the compatible child $\ell$ of $i$. Removing $i$ may add new intervals to the equivalence class of $\ell$. Therefore if $\ell$ is still the greatest interval in the updated equivalence class, we add an edge $(\ell, \mathsf{rc}(\ell)$ to $E_c$. Otherwise, we add the edge $(i,j)$ to $E_{\sim}$, where $j$ is the next interval of~$\ell$.

\begin{algorithm}[H]
	\caption{$\mathsf{removeLT}(i)$}\label{alg:LT-remove}
	\begin{algorithmic}[1]
	\If{$i$ is not the root}
		\State $\mathsf{cut}(i)$	
	\EndIf	
	\State $j \leftarrow \mathsf{previous}(i)$
        \If{$\mathsf{rc}(j)=\mathsf{rc}(i)$}		\Comment $(j, i) \in E_\sim$ %is the equivalent child of $i$.
        		\State $\mathsf{cut}(j)$
		        \If{$i = \mathsf{lc}(\mathsf{rc}(i))$}					\Comment $\mathsf{rc}(i)$ is a new parent of $j$      					
		        	\State $\mathsf{link}(j,\mathsf{rc}(i))$			
						\ElsIf{$i$ is not the the root } 								\Comment $\mathsf{next}(i)$ is a new parent of $j$
							\State $\mathsf{link}(j, \mathsf{next}(i))$		
      			\EndIf
	\EndIf
        \State $j \leftarrow \mathsf{lc}(i)$
        \If{$i =\mathsf{rc}(j)$}							\Comment $(j, i) \in E_c$
            \State $\mathsf{cut}(j)$
            \State $\mathsf{remove}(T(I),i)$
            \State $k \leftarrow \mathsf{next}(j)$ %\Comment $k$ is the next interval of $j$ w.r.t. $I' = I \setminus \{i\}$
            \If{$j$ is not the root }
	    \If{$\mathsf{rc}(k) = \mathsf{rc}(j)$}  \Comment $\mathsf{rc}(j) \ne i$ as we removed $i$ from $T(I)$.
        		\State $\mathsf{link}(j,k)$
	    \Else
	    	\State $\mathsf{link}(j,\mathsf{rc}(j))$
             \EndIf
             \EndIf
	  \Else \State $\mathsf{remove}(T(I),i)$
        \EndIf
	\end{algorithmic}
\end{algorithm}

\subsection{Correctness of the update operations}

To prove correctness of the algorithms above, we state two claims about linearised trees. The first claim allows us to check if the given interval the greatest in its equivalent class. The second claim says that changes of the linearised tree after insertion or deletion of an interval $i$ are local with respect to~$i$. We abuse notation and write $(I; E)$ instead of $(I; E_c, E_\sim)$ and $(I'; E')$ instead of~$(I'; E'_c, E'_\sim)$. Which edges are used will be clear from the context.

\begin{claim}\label{lem:greatest}
An interval $i$ is the greatest in $[i]$ if and only if $\mathsf{lc}(\mathsf{rc}(i))=i$.
\end{claim}
\begin{proof}
Let $i$ be the greatest interval in $[i]$. Then for any $j \in [i]$ we have that $j \preceq i$. Assume that $k = \mathsf{lc}(\mathsf{rc}(i))\ne i$. Then $i \preceq k$ which is a contradiction.

For the other direction, assume that $i$ is not the greatest interval in its equivalent class, that is there exists $j \in [i]$ such that $i \prec j$. Clearly, $j$ is compatible with $\mathsf{rc}(i)$. Therefore $\mathsf{lc}(\mathsf{rc}(i))=j$, witch is a contradiction. \qed
\end{proof}

\begin{claim}\label{lem:diffLT}
Let ${\cal L}(I) = (I, E)$ and ${\cal L}(I') = (I', E')$ be two linearised trees such that $I' = I\cup\{i\}$. Let $j$ and $k$ be intervals from the set $I'$. Then the following properties are satisfied:
\begin{itemize}
\item[(1)] if $(j,k) \not \in E$ and $(j,k) \in E'$, then either $j = i$ or $k = i$.
\item[(2)] if $(j,k) \in E$ and $(j,k) \not \in E'$, then $(j,i) \in E'$.
\end{itemize}
\end{claim}
\begin{proof}
For the first property, we note that if two intervals from $I$ are not connected by an edge in ${\cal L}(I)$ then they are not connected by an edge in a bigger linearised tree ${\cal L}(I')$. Hence either $j=i$ or $k=i$.
%Proposition~\ref{lem:rootLT} implies
%It is clear that the root of the linearised tree $\mathcal{L}(I)$ is the greatest interval in $I$. Therefore, as $j$ has a parent, $j$ is not the root in $\mathcal{L}(I \cup \{i\})$.
For the second property, let $\ell \neq k$ be a parent of $j$ in $\mathcal{L}(I')$. Because $(j, \ell) \not \in E$ and $(j, \ell) \in E'$, the property (1)  implies that either $j=i$ or $\ell = i$. Thus~$\ell = i$.
\qed
\end{proof}

\begin{lemma}\label{lem:insert-LT}
The operation $\mathsf{insertLT}(i)$ preserves linearised tree data structure.
\end{lemma}
\begin{proof}
Consider intervals $j, k \in I'$, where $j\preceq k$ and $I' = I \cup \{i\}$. Let $(I',E')$ be the resulting tree after
the algorithm $\mathsf{insertLT}(i)$ is performed. We show that $(j,k) \in E'$ if and only if $(j,k)$ is an edge in $\mathcal{L}(I')$.

\smallskip

\noindent
$(\rightarrow)$ Suppose that $(j,k) \in E'$. We prove that $(j,k)$ is an edge in $\mathcal{L}(I')$.
\begin{itemize}
	\item Let $(j,k) \not \in E$. Then the algorithm $\mathsf{insertLT}$ must have added $(j,k)$ into $E'$. Any edge the algorithm adds is adjacent to 	$i$. First, we consider outgoing edges, that is, we consider the case when $j=i$. If the algorithm adds an edge from $i$ to $\mathsf{rc}(i)$, then $i = \mathsf{lc}(\mathsf{rc}(i))$ (see lines 3-4 of the Algorithm~\ref{alg:LT-insert}). By Sublemma~\ref{lem:greatest}, $i$ is the greatest interval in its equivalence class. If the algorithm adds an edge from $i$ to the next interval $k$ of $i$, then $i$ is not the greatest interval in $[i]$ and $k \sim i$ (see lines 3-6). Second, we consider incoming edges, that is, we consider the case when $k=i$. If the algorithm adds an edge from $\mathsf{lc}(i)$ to $i$, then $j = \mathsf{lc}(\mathsf{rc}(j))$ (see lines 10-12). By Sublemma~\ref{lem:greatest}, $j$ is the greatest interval in its equivalence class. If the algorithm adds an edge from the previous interval $j$ of $i$ to $i$, then $j \sim i$ (see lines 7-9). Note that any of the edges added by the algorithm is an edge in $\mathcal{L}(I')$. Hence $(j,k)$ is an edge in $\mathcal{L}(I')$.
	
	\item Let $(j,k) \in E$. Assume that $(j,k)$ is not an edge in $\mathcal{L}(I')$. By Sublemma~\ref{lem:diffLT} $(j,i)$ is an edge in $\mathcal{L}(I')$. If $j$ is the equivalent child of $i$, then $j$ is the previous interval of $i$ and $\mathsf{rc}(j) = \mathsf{rc}(i)$. If $j$ is the compatible child of $i$, then $j = \mathsf{lc}(\mathsf{rc}(j))$. In both of these cases the algorithm deletes the edge from $j$ (see lines 7-9 and 10-12 correspondently). Thus $(j,k) \not \in E'$, which is a contradiction.
	\end{itemize}

\noindent
$(\leftarrow)$ Suppose that $(j,k)$  is an edge in $\mathcal{L}(I')$. We prove that $(j,k) \in E'$.
\begin{itemize}
  \item Let $(j,k) \in E$. Assume that $(j,k) \not \in E'$. Then the algorithm $\mathsf{insertLT}$ must have deleted $(j,k)$. There are two cases: $j$ is the previous interval of $i$ and $j \sim i$  (see lines 7-9), or $i = \mathsf{rc}(j)$ and $j$ is the greatest interval in $[j]$ (see lines 10-12). In either case, $j$ is a child of $i$ in $\mathcal{L}(I')$, that is $k=i$, which is a contradiction to the assumption that $(j,k)\in E$.

  \item Let $(j,k) \not \in E$. By Sublemma~\ref{lem:diffLT}, either $j = i$ or $k=i$. Suppose $j = i$. If $i$ is the compatible child of $k$ in $\mathcal{L}(I')$, then $k = \mathsf{rc}(i)$ and, by Sublemma~\ref{lem:greatest}, $i = \mathsf{lc}(\mathsf{rc}(i))$. If $i$ is the equivalent child of $k$, then $k$ is the next interval $i$ and $k \sim i$. The algorithm $\mathsf{insertLT}$ adds the edge $(i,k)$ to $E'$ in lines 3-6. Suppose, $k = i$. If $j$ is the equivalent child of $i$, $j$ is the previous interval of $i$ and $j \sim i$. If $j$ is the compatible child of $i$, then $j = \mathsf{lc}(\mathsf{rc}(j))$. The algorithm $\mathsf{insertLT}$ adds the edges $(j,i)$ to $E'$ in lines 7-12. In any case, the edge $(j,k) \in E'$.
  \qed
    \end{itemize}

\end{proof}

\begin{lemma}\label{lem:remove-LT}
The operation $\mathsf{removeLT}(i)$ preserves  linearised tree data structure.
\end{lemma}
\begin{proof} Suppose $\mathcal{L}(I)=(I,E)$ is the linearised tree of a set $I$ of intervals and $(I \setminus \{i\},E')$ is the resulting tree after
the algorithm $\mathsf{removeLT}(i)$ is performed. Consider intervals $j$ and $k$ in $I$, where $j\preceq k$. We want to show that $(j,k) \in E'$ if and only if $(j,k)$ is an edge in $\mathcal{L}(I \setminus \{i\})$.
%This will proof that the tree $(I \setminus \{i\},E')$ is the linearised tree $\mathcal{L}(I \setminus \{i\})$.

\smallskip
\noindent
$(\rightarrow)$ Suppose $(j,k) \in E'$. We prove that $(j,k)$ is an edge in $\mathcal{L}(I \setminus \{i\})$.
\begin{itemize}
\item Let $(j,k) \in E$. Assume that $(j,k)$ is not an edge in $\mathcal{L}(I \setminus \{i\})$. By Lemma~\ref{lem:diffLT}, either $j = i$ or $k = i$. If $j=i$, that is, $i$ is a child of $k$ in $\mathcal{L}(I)$, then the algorithm removes the edge $(i,k)$ in line 2. Consider the case when $k=i$, that is, $j$ is a child of $i$. If $j$ is the equivalent child of $i$, the algorithm removes the edge $(j,i)$ in lines 3-5. If $j$ is the compatible child of $i$, the algorithm removes $(j,i)$ in lines 10-12. In either case $(j,k) \not \in E'$, which is a contradiction.

\item Let $(j,k) \not \in E$. The algorithm $\mathsf{removeLT}$ must have added the edge $(j,k)$. There are four possible cases. First, the algorithm adds an edge in line 7, that is, $k = \mathsf{rc}(i)$. Then $j$ is the equivalent child of $i$ and $i$ is the greatest interval in $[i]$. After removing $i$, $j$ is the greatest interval in $[j]$, so that $j$ is the compatible child of $k$. Second, the algorithm adds an edge in line 9, that is, $k$ is the next interval of $i$. Then $i \sim k$. Since $j$ is the equivalent child of $i$, $j \sim k$. Third, the algorithm adds an edge in line 17. Then $k$ is the next interval of $j$ with respect to $I \setminus \{i\}$ and $j \sim k$. Finally, the algorithm adds an edge in line 19.  Then $j$ is the greatest interval in $[j]$  and $k = \mathsf{rc}(j)$ with respect to $I \setminus \{i\}$. In all these case the edge $(j,k)$ is an edge in $\mathcal{L}(I \setminus \{i\})$.

\end{itemize}

\noindent
$(\leftarrow)$ Suppose $(j,k)$ is an edge in $\mathcal{L}(I \setminus \{i\})$. We prove that $(j,k) \in E'$.
\begin{itemize}
\item Let $(j,k) \in E$. Assume that $(j,k) \not \in E'$. Then the algorithm $\mathsf{removeLT}$ must have deleted the edge $(j,k)$. First, the algorithm removes an edge from $i$ (see line 2). Second, it removes an edge from the equivalent child of $i$ (see lines 3-58). Finally, it removes an edge from the compatible child of $i$ (see lines 10-12). Thus the algorithms removes only edges, incident to $i$, but these edges are not in $\mathcal{L}(I \setminus \{i\})$, which is a contradiction.

\item Let $(j,k) \not \in E$. By Lemma~\ref{lem:diffLT} $(j,i)$ is an edge in $\mathcal{L}(I)$.
Suppose $j$ is the equivalent child of $i$. The algorithm finds $j$ in lines 3-5. If $j$ is the compatible child of $k$ in $\mathcal{L}(I \setminus \{i\})$, then $i$ is the compatible child of $k$ in $\mathcal{L}(I)$. If $j$ is the equivalent child of $k$, then $k$ is the next interval of $i$. The algorithm takes care of both cases in lines 6-9 and adds the edge $(j,k)$ in line~9.
Suppose $j$ is the compatible child of $i$. The algorithm finds $j$ in lines 10-11. If $j$ is the equivalent child of $k$, then $k$ is the next interval of $j$ and $\mathsf{rc}(j) = \mathsf{rc}(k)$ with respect to $I \setminus \{i\}$. The algorithm adds the edge $(j,k)$ in lines 15-17. If $j$ is the compatible child of $k$, then the algorithm adds the edge in line 19. Thus, $(j,k) \in E'$.\qed
\end{itemize}
\end{proof}

Lemmas~\ref{lem:query-LT}-\ref{lem:remove-LT}  lead us to the following theorem:

\begin{theorem}\label{thm:LT-amortized-time}
The  $\mathsf{queryLT}$, $\mathsf{insertLT}$ and $\mathsf{removeLT}$ operations
solve the dynamic monotonic interval scheduling problem in $O(\log n)$ amortised time,
where $n$ is the size of the set $I$ of intervals.
\end{theorem}

\noindent {\bf Note}. The time complexity of the operations above depends on the type of dynamic trees, representing paths of~$\mathsf{LT}$. We can achieve the worst-case bound instead of amortized if we use globally biased trees instead of splay trees~\cite{tarjan83}. However, after each operation we must ensure that for every pair of edges $(v, u)$ and $(w, u)$ of the linearised tree, nodes $v$ and $u$ are in the same dynamic tree if and only if the numbers of nodes in the subtree rooter at $v$ is greater or equal to the number of nodes in the subtree rooted at~$u$.
%~$2 \cdot \mathsf{size} (v) \geq \mathsf{size}(w)$.

\section{Experimental results}\label{sec:experiment}
%%%%%%%%%%%%%%%%%%%%%%%%%%%%%%%%%%%%%%%%%%

\graphicspath{
	{../old/plots/experiments/1/graph1/},
	{../old/plots/experiments/1/graph2/},
	{../old/plots/experiments/2/graph1/},
	{../old/plots/experiments/2/graph2/}
}

In this section we present an experimental comparison between three algorithms for solving monotonic case of the dynamic interval scheduling problem: (i) the naive dynamic algorithm $\mathsf{N}$, which keeps the intervals a binary search and calculate the greedy optimal set from scratch at every query operation; (ii) the algorithm $\mathsf{CF}$ based on the compatibility forest; (iii) the algorithm $\mathsf{LT}$ based on the linearised tree. We implemented these algorithms in Java. The algorithm $\mathsf{N}$ is based on the standard Java implementation of Red-Black tree, which we extended with  ${\mathsf{left\_compatible}}$ and ${\mathsf{right\_compatible}}$ operations. We use the implementation of $\mathsf{N}$ in the algorithms $\mathsf{CF}$ and $\mathsf{LT}$ to store intervals and perform tree operations. In $\mathsf{CF}$ and $\mathsf{LT}$ we implemented bottom-up splay operation as described in \cite{tarjan83}. We run the experiments on a laptop with {\em 4GB of RAM} memory and {\em Intel Core 2 Duo 2130 Mhz, 3MB of L2 cache memory} processor.

In our experiments, we measure the total and the average running time of a sequence of $m$ operations on initially empty interval set. The sequence consists of $n$ insert operations, $rn$ remove operations and $qn$ query operation, where $n$ is a linearly increasing number and $r$ and $q$ are fixed parameters of the experiment. We create a sequence of operations randomly while satisfying two conditions. First,  whenever we invoke an insert operation of an interval $i$, we make sure that there is no interval $i$ in the set. Second, whenever we invoke a remove operation of $i$, we make sure that $i$ exists in the set. Thus every update operation calls for an actual change of the interval set.

To better understand the algorithms' performance, we defined the {\em sparsity} of an interval set $I$ to be the upper bound on the ratio between the size of the greedy optimal set $J$ and the size of $I$. The smaller the sparsity, the more intervals pairwise overlap. For example, if the sparsity is $1/2$, we make sure by creating intervals of the length $2/n$ that at most every second interval can belong to $J$.

The sparsity of $I$ has an important influence on the algorithms $\mathsf{N}$ and $\mathsf{CF}$. In the compatibility forest we conclude every update operation with the expose operation on the least interval in the set, which restores the missing edges between intervals from $J$. Therefore the smaller sparsity, the smaller chance of an update operation to affects the splay tree, representing set $J$. In the naive algorithm, the query operation may visit every interval from $J$. Therefore the smaller sparsity, the less maximal number of intervals the query operation may visit.

\smallskip

\noindent {\em Experiment 1.} The analysis of the algorithms shows that $\mathsf{N}$ updates the interval set faster than $\mathsf{CF}$ and $\mathsf{LT}$, but queries the set slower. Therefore in the first experiment we measured the efficiency of the algorithms undergoing $n$ insert, $0.5n$ remove and $0.01n$ query operations. The operations are shuffled as described above. We set the sparsity parameter to be $0.1$. The result of the experiment is shown on the Figure \ref{fig:slownaive}.

\begin{figure}[htb]	
	\includegraphics[] {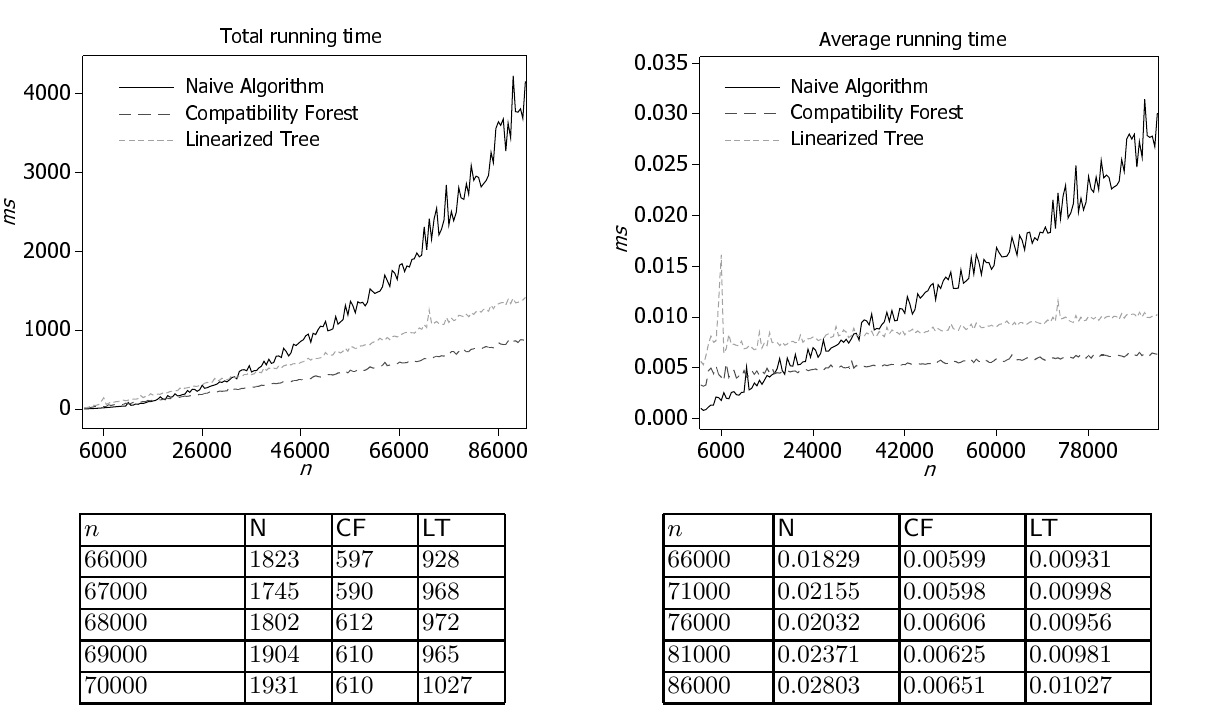}
	\caption{ sparsity is $0.1$, $0.5n$ remove operations, $0.01n$ query operations, $n$ insert operations.}
	\label{fig:slownaive}

\end{figure}

The experiment shows that the difference of total running time between algorithms undergoing a sequence of operations with number of insertion less then 6000 is small, especially between $\mathsf{CF}$ and $\mathsf{N}$. However, when we increase the number of insert operations, $\mathsf{N}$ performs much slower than two other algorithms. The average running time per operation of $\mathsf{N}$ is increasing similarly to a linear function, whereas the average running time per operation of $\mathsf{CF}$ and $\mathsf{LT}$ increases much slower. The experiment also shows that $\mathsf{CF}$ updates the interval set with low sparsity faster than $\mathsf{LT}$.

\smallskip

\noindent {\em Experiments 2 and 3}. In the next two experiments we measure the performance of $\mathsf{CF}$ and $\mathsf{LT}$ undergoing a sequence of operations with the equal number of insert and query operations. We excluded $\mathsf{N}$ from the experiments because $\mathsf{N}$ performs too slowly when the number of query operations increases. The difference between the second and the third experiment is in the number of remove operations. Sequences in Experiment~2 do not contain remove operations. Sequences in Experiment~3 contain $0.5n$ remove operations. We set the sparsity parameter to be $0.8$. Figure \ref{fig:badlt} shows the results of Experiment~2, Figure \ref{fig:similar} shows the results of Experiment~3.

\begin{figure}[htb]
	\includegraphics[] {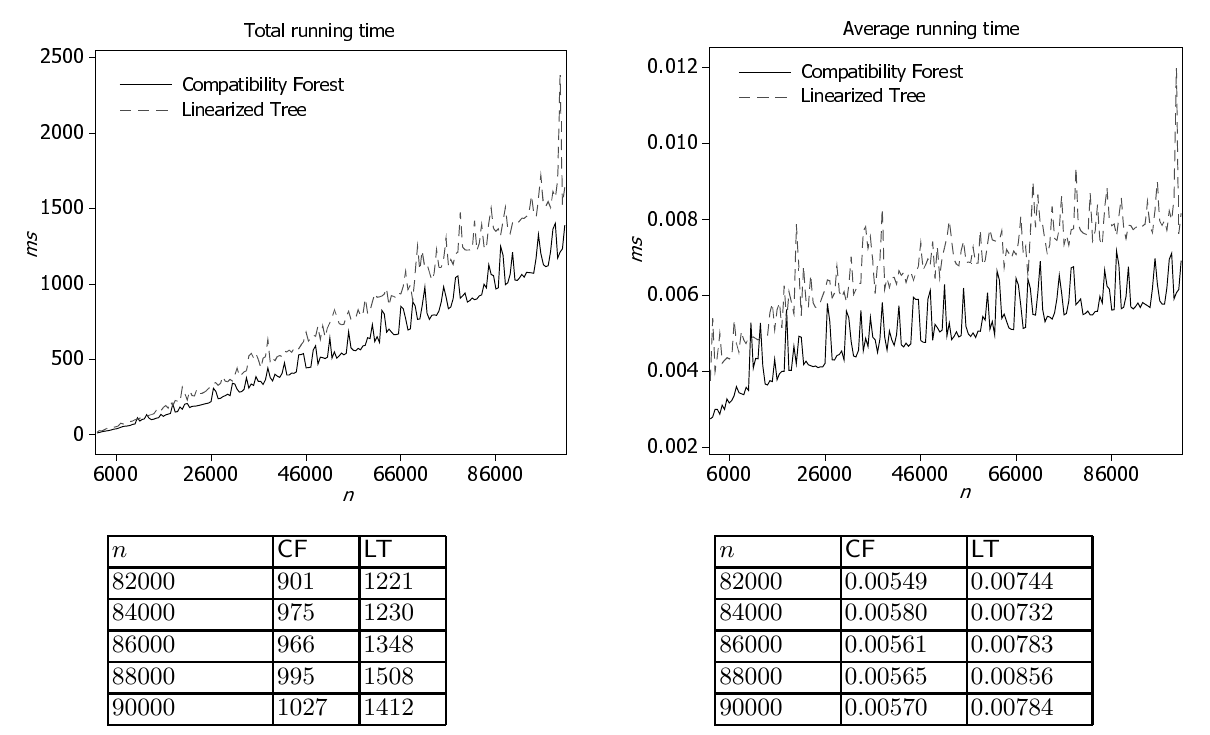}
	\caption{ sparsity is $0.8$, no remove operations, $n$ query and insert operations.}
	\label{fig:badlt}
\end{figure}

\begin{figure}[htb]
	\includegraphics[] {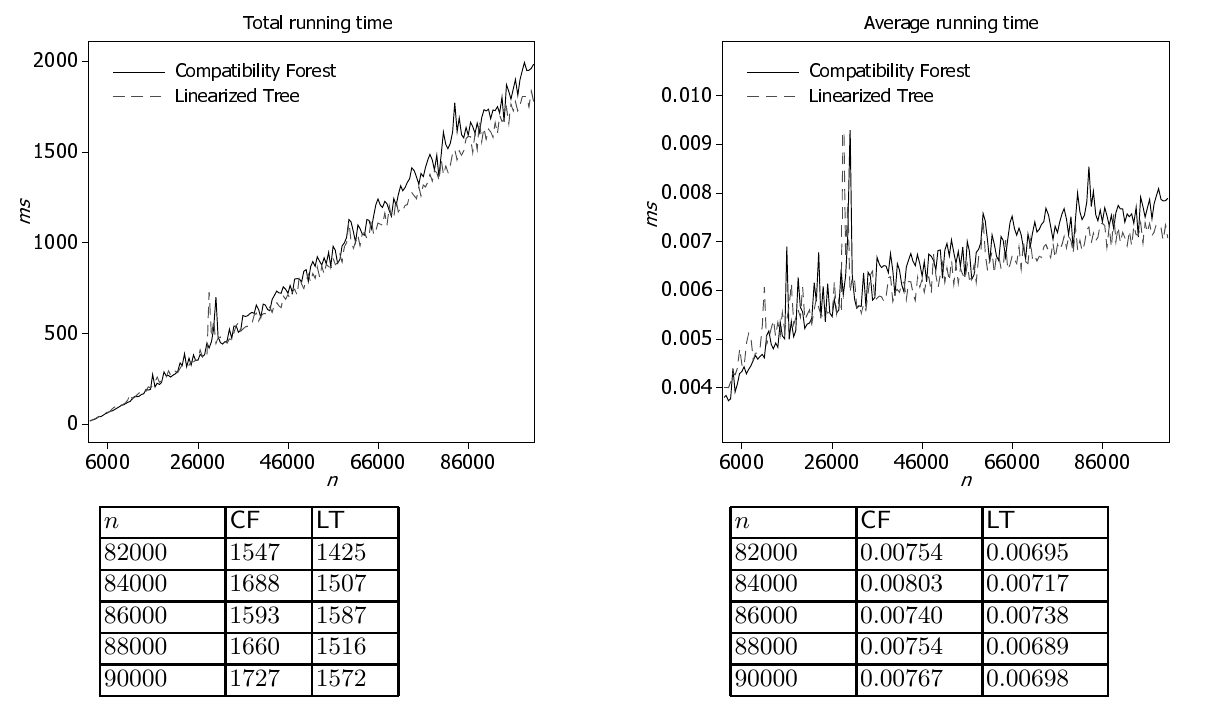}
	\caption{ sparsity is $0.8$, $0.5n$ remove operations, $n$ query and insert operations.}
	\label{fig:similar}
\end{figure}

The third experiment shows that if we do not allow remove operations, $\mathsf{CF}$  performs faster than $\mathsf{LT}$. If we allow remove operations, $\mathsf{CF}$ performs slightly slower than $\mathsf{LT}$. However, the results of Experiment~2 show that if the interval set is not sparse, $\mathsf{CF}$ inserts and removes intervals faster than $\mathsf{LT}$.

\smallskip

{\em Conclusion}. The experimental result verifies our theoretical analysis and shows that both $\mathsf{CF}$ and $\mathsf{LT}$ runs significantly faster than the naive algorithm. Moreover, the results show that in a random environment $\mathsf{CF}$ performs as fast as $\mathsf{LT}$ to within a constant factor, despite the worst ($\log^{2} n$) time upper bound. Considering that $\mathsf{CF}$ is relatively easy to implement, $\mathsf{CF}$ can find its practical applications.


\newpage



\thebibliography{99}

%\bibitem{toptrees}
%	S.\,Alstrup, J.\,Holm, K.D.\,Lichtenberg and M.\,Thorup.
%	Maintaining information in fully dynamic trees with top trees.
%	{\em ACM Transactions on Algorithms}, 2005, 1(2), 243-264.

\bibitem{piglinear2004}
  D.\,Corneil.
  A simple 3-sweep LBFS algorithm for the recognition of unit interval graphs.
  {\em Discrete Applied Mathematics}, 2004, 138(3), 371-379.

%\bibitem{cormen}
%  T.\,Cormen, C.\,Leiserson, R.\,Rivest, and C.\,Stein,
%  {\em Introduction to Algorithms}, 2001.

\bibitem{piglinear1996}
  X.\,Deng, P.\,Hell and J.\,Huang.
  Linear-time representation algorithms for proper circular-arc graphs and proper interval graphs.
  {\em SIAM Journal on Computing}, 1996, 25(2), 390-403.

\bibitem{fpz}
	S.\,Fung, C.\,Poon, and F.\,Zheng.
	Online interval scheduling: Randomized and Multiprocessor cases.
	{\em Proceedings of COCOON}, 176-186, 2007.

\bibitem{pset}
	M.J.\,Katz, F.\,Nielsen and M.\,Segal.
	Maintenance of a piercing set for intervals with applications.
	{\em Algorithmica}, 36(1), 2003, 59-73

\bibitem{lt}
	R.\,Lipton and A.\,Tompkins.
	Online interval scheduling.
	{\em Proceedings of the Fifth Annual ACM-SIAM Symposium on Discrete Algorithms}, 1994, 302-311.

\bibitem{newpig}
  P.\,Heggernes, D.\,Meister, and C.\,Papadopoulos.
  A new representation of proper interval graphs with an application to clique-width.
  {\em Electronic Notes in Discrete Mathematics}, 2009, 32, 27-34.

\bibitem{nested}
  H.\,Kaplan,E.\,Molad and R.\,Tarjan.
  Dynamic rectangular intersection with priorities.
  {\em Proceedings of the thirty-fifth annual ACM symposium on Theory of computing}, 2003, June, 639-648.

\bibitem{kleinberg}
  J.\,Kleinberg and E.\,Tardos. {\em Algorithm Design}, 2006.

\bibitem{survey07}
 A.\,Kolen, J.K.\,Lenstra, C.H.\,Papadimitriou, and F.C.\,Spieksma.
 Interval scheduling: A survey, {\em Naval Research Logistics}, 54, 5, 2007, 530--543.

\bibitem{i_tree}
Mehlhorn, K. {\em Data structures and algorithms, Volume 3: Multi-dimensional Searching and Computational Geometry}. Springer-Verlag, Berlin, 1984.

\bibitem{tarjan83}
  D.\,Sleator and R.\,Tarjan.
  A Data Structure for Dynamic Trees,
  {\em Journal of computer and system sciences}, 26, 3, 1983, 362--391.

\bibitem{tarjan85}
  D.\,Sleator and R.\,Tarjan.
  Self-adjusting binary search trees,
  {\em Journal of the ACM}, 32, 3, 1985, 652--686.
\end{document}